\documentclass[journal, 12pt, draftclsnofoot, onecolumn]{IEEEtran}
\ifCLASSINFOpdf
\else
\fi

\newtheorem{prop}{Proposition}

\newcommand{\be}{\begin{equation}}
\newcommand{\ee}{\end{equation}}
\newcommand{\ba}{\begin{array}}
\newcommand{\ea}{\end{array}}
\newcommand{\bea}{\begin{eqnarray}}
\newcommand{\eea}{\end{eqnarray}}

\newcommand{\vbar}{\raisebox{.17ex}{\rule{.04em}{1.35ex}}}
\newcommand{\vbarind}{\raisebox{.01ex}{\rule{.04em}{1.1ex}}}
\newcommand{\R}{\ifmmode {\rm I}\hspace{-.2em}{\rm R} \else ${\rm I}\hspace{-.2em}{\rm R}$ \fi}
\newcommand{\T}{\ifmmode {\rm I}\hspace{-.2em}{\rm T} \else ${\rm I}\hspace{-.2em}{\rm T}$ \fi}
\newcommand{\N}{\ifmmode {\rm I}\hspace{-.2em}{\rm N} \else \mbox{${\rm I}\hspace{-.2em}{\rm N}$} \fi}
\newcommand{\B}{\ifmmode {\rm I}\hspace{-.2em}{\rm B} \else \mbox{${\rm I}\hspace{-.2em}{\rm B}$} \fi}
\newcommand{\Hil}{\ifmmode {\rm I}\hspace{-.2em}{\rm H} \else \mbox{${\rm I}\hspace{-.2em}{\rm H}$} \fi}
\newcommand{\C}{\ifmmode \hspace{.2em}\vbar\hspace{-.31em}{\rm C} \else \mbox{$\hspace{.2em}\vbar\hspace{-.31em}{\rm C}$} \fi}
\newcommand{\Cind}{\ifmmode \hspace{.2em}\vbarind\hspace{-.25em}{\rm C} \else \mbox{$\hspace{.2em}\vbarind\hspace{-.25em}{\rm C}$} \fi}
\newcommand{\Q}{\ifmmode \hspace{.2em}\vbar\hspace{-.31em}{\rm Q} \else \mbox{$\hspace{.2em}\vbar\hspace{-.31em}{\rm Q}$} \fi}
\newcommand{\Z}{\ifmmode {\rm Z}\hspace{-.28em}{\rm Z} \else ${\rm Z}\hspace{-.28em}{\rm Z}$ \fi}

\usepackage{graphics} % for pdf, bitmapped graphics files
\usepackage{epsfig} % for postscript graphics files
\usepackage{times} % assumes new font selection scheme installed
\usepackage{amsmath} % assumes amsmath package installed
\usepackage{amssymb}  % assumes amsmath package installed
\usepackage{mathrsfs}
\usepackage{array}
\usepackage{algorithm}
\usepackage{algorithmic}
\usepackage{changepage}
\usepackage{lipsum}
\usepackage{multicol}
\usepackage{caption}

\usepackage{url}
\usepackage{array}

\usepackage[english]{babel}
\usepackage{times}
\usepackage{psfrag}
\usepackage{dsfont}
\usepackage{algorithmic}
\usepackage{algorithm}
\usepackage{setspace}
\usepackage{amssymb}
\usepackage{amsmath}
\usepackage{cite}
\usepackage{graphicx}
\usepackage{color}
\usepackage{xcolor}
\usepackage{subfigure}
\usepackage{epsfig}
\usepackage{stfloats}
\usepackage{microtype}
\DisableLigatures{encoding=*,family=*}

\def\R{\mathds{R}}
\def\N{\mathds{N}}

\begin{document}
\title{Auction-based Resource Allocation in \\MillimeterWave Wireless Access Networks}
\author{
        George Athanasiou,~\IEEEmembership{Member,~IEEE,}
        Pradeep Chathuranga Weeraddana,~\IEEEmembership{Member,~IEEE},
        and Carlo Fischione,~\IEEEmembership{Member,~IEEE,}

\thanks{The authors are with Electrical Engineering School, Access Linnaeus Center, KTH Royal Institute of Technology, Stockholm, Sweden. E-mail:\textit{ \{georgioa, chatw, carlofi\}@kth.se}. This work was supported by the Swedish Research Council and the EU project Hydrobionets.}}
\maketitle
\vspace{-0.3cm}
\begin{abstract}
The resource allocation problem of optimal assignment of the clients to the available
access points in 60 GHz millimeterWave wireless access networks is investigated. The problem
is posed as a multi-assignment optimization problem.
The proposed solution method converts the initial problem to a minimum cost flow problem and allows to design an efficient algorithm by
a combination of auction algorithms. The solution algorithm exploits the network optimization structure of the problem, and thus is much more powerful than computationally intensive general-purpose solvers. Theoretical and numerical results evince numerous properties, such as optimality, convergence, and scalability in comparison to existing approaches.

\end{abstract}

\IEEEpeerreviewmaketitle

\vspace{-0.3cm}
\section{Introduction}\vspace{-0.2cm}
MillimetterWave (mmW) communications utilize the part of the electromagnetic spectrum between 30 and 300~GHz, which corresponds to wavelengths from 10~mm to 1~mm~\cite{Smulders07}. MmW wireless networks in the 60 GHz unlicensed band are considered one of the key technologies for enabling multi-gigabit wireless access (transmission rates up to 7 Gbps) and provisioning of QoS-sensitive applications. Multiple industry-led efforts and international organizations have emerged for the standardizationc. More than 5~GHz of continuous bandwidth is available in many countries worldwide, which makes 60~GHz systems particularly attractive for gigabit wireless applications such as gigabyte file transfer, wireless docking station, wireless gigabit ethernet, wireless gaming, and uncompressed high definition video transmission. Moreover, scenarios such as dense small-cells and mobile data offloading~\cite{Lee12}, which are nowadays strongly motivated by the increased end-user connectivity requirements and mobile traffic, can be accommodated with the use of 60~GHz radio access technology.

Resource allocation for wireless local area networks has been the focus of intense research. Several studies have analyzed the performance of the basic client association policy that IEEE 802.11 standard defines, based on the received signal strength indicator (RSSI). These studies have showed that this basic association policy can lead to inefficient use of the network resources \cite{Bejerano2}. Therefore, there has been increasing interest in designing better client association policies \cite{Athanasiou07, Athanasiou09, Shakkottai, Chen13}. Whereas the previous approaches are hard to apply in 60 GHz wireless access networks due to the special characteristics of the 60 GHz channel, and the differences with the rest wireless access technologies \cite{Qiao11, Singh11, Genc12, Lin12} (namely, severe channel attenuations, high path loss, directionality, and blockage), novel mechanisms must be designed to provide optimal resource allocation. Our previous approach \cite{Athanasiou-etal-2013} was the first to study the client association in 60 GHz wireless access networks. However, the focus was on the network performance (achieving load balancing) and not on optimizing the benefit of the individual clients.

This paper considers the special characteristics of the 60 GHz access channel and poses the client association optimization problem, where the objective is to \emph{maximize the total clients benefit} in the network. To address the problem, we propose an iterative approach that combines two auction algorithms. We compare our solution method to basic association policies, already in use in the present 60 GHz communication technologies under standardization (802.15.3c, 802.11ad)~\cite{802_11ad}.

The rest of the paper is organized as follows. A description of the system model and the problem formulation is presented in \S~\ref{sec:SysModel_mini_max_primal_problem}. In \S~\ref{sec:sol_prob}, we describe the solution approach to the multi-assignment problem. In \S~\ref{sec:numerical_results} numerical results are presented. Lastly, \S~\ref{sec:conclusions} concludes the paper.

\vspace{-0.5cm}
\section{System Model and Problem Formulation}\label{sec:SysModel_mini_max_primal_problem}\vspace{-0.2cm}

We consider a mmW network where $m$ access points (APs) that can serve $n$ clients and $n\geq m$. An AP $i$ can serve more than one client. Moreover, every client $j$ must be associated to just one AP.
%In \textit{multiassignment problems} we have a set of $m$ persons that must be assigned to $n$ objects. In general, $n \geq m$. The benefit that is achieved for the assignment of a person $i$ to the object $j$ is $a_{ij}$.
The set of clients to which AP $i$ can be assigned is a nonempty set $A(i)$. Moreover, we introduce the set $B(j)$ as the nonempty set of APs that can serve client $j$. An assignment $S$ is defined as a set of AP-client pairs $(i, j)$, with $j \in A(i)$, where each AP $i$ can be part of more than one pair $(i,j) \in S$, and where every client $j$ must be part of only one pair $(i,j) \in S$. %An assignment is \textit{feasible} if it contains $n$ pairs, where each client $j$ is assigned to a person. %The general objective is to find an assignment that maximizes the total benefit $\sum\limits_{(i,j) \in A} a_{ij}$.
%A mmW wireless network consisting of $N$ APs and $M$ clients is considered. We denote by $\mathcal{N}$ the set of APs and by $\mathcal{M}$ the set of clients, where $\mathcal{N}=\{1,\ldots,N\}$ and $\mathcal{M}=\{1,\ldots,M\}$.
%%
%The set of clients that can be associated to AP~$i$ is denoted by $\mathcal{M}_i$.
%%We assume that there are no isolated clients, i.e., $\mathcal{M}_1\cup\cdots\cup\mathcal{M}_N=\mathcal{M}$.
%We denote by $\mathcal{N}_j$ the set of candidate APs that client $j$ could be associated with.
An illustrative example of access network is shown in Figure~\ref{fig:system_example}, where the clients positioned inside a disc with radius $r$ centered at the location of AP~$i$ can be associated with that AP.

\begin{figure}[t]
\centering
\includegraphics[height=0.23\textheight]{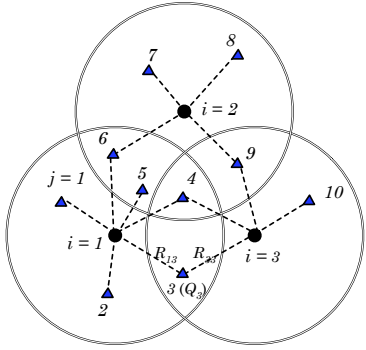}\vspace{-5mm}
\caption{Example mmW wireless access network.}
\label{fig:system_example}
\vspace{-0.8cm}
\end{figure}

Every node is equipped with steerable directional antennas and it can direct its beams to transmit or to receive \cite{Qiao11}. We assume that AP~$i$ can support its clients with a separate transmit beam. We consider the case where all receiver nodes are using single-user detection (i.e., a receiver decodes each of its intended signals by treating all other interfering signals as noise) and assume that the achievable rate from AP~$i$ to client~$j \in A(i)$ is
\vspace{-2mm}\begin{equation} \label{eq:AWGN}
R_{ij} = W\log_2\bigg (1+\frac{P_{ij}G_{ij}}{(N_0+I_{j})W}\bigg)\ ,
\vspace{-2mm}\end{equation}
where $W$ is the system bandwidth, $P_{ij}$ is the transmission power of AP~$i$ to client~$j$, $G_{ij}$ is the power gain from AP~$i$ to client~$j$, $N_0$ is the power spectral density of the noise at each receiver, and $I_{j}$ is the interference spectral density at client~$j$.  We use the Friis transmission equation together with the~\emph{flat-top} transmit/recieve antenna gain model~\cite{Singh11}, where a fixed gain is considered within the beamwidth and zero gain is considered outside the beamwidth of the antenna. %In addition, we consider Rayleigh small-scale fading.
%Thus we have
%\be\label{eq:gain_coefficient}
%G_{ij} = \frac{G^{\mathrm{Tx}}_{ij}G^{\mathrm{Rx}}_{ij}\lambda^2\alpha_{ij}}{16\pi^2(d_{ij}/d_0)^\eta} \ , \ i\in\mathcal{N}, j\in\mathcal{M}_i \ ,
%\ee
%where $G^{\mathrm{Tx}}_{ij}$ is the transmit antenna gain from AP~$i$ to client~$j$, $G^{\mathrm{Rx}}_{ij}$ is the receive antenna gain from AP~$i$ to client~$j$, $\lambda$ is the wavelength, $\alpha_{ij}$ is the \emph{fading coefficient} which is an exponentially distributed random variables with unit mean to models the Rayleigh small-scale fading, $d_{ij}$ is the distance between AP $i$ and client $j$, $d_0$ is the \emph{far field reference distance} \cite{Anurag-Manjunath-Kuri-08}, $\eta$ is the path loss exponent,\footnote{$\eta\in [2,6]$ in IEEE~802.11ad networks~\cite{Geng09}.} and $I_{j}$ is the communication interference at client~$j$.
%
We capitalize on the well studied 60~GHz propagation characteristics~\cite{Singh11}, such as highly directional transmissions with very narrow beamwidths and increased path losses due to the oxygen absorption, in order to assume that the communication interference $I_j$ is very small and does not affect significantly the achievable communication rates in the network. We remark that all the assumptions that we make above are natural for 60~GHz~\cite{Singh11}.

We denote by $Q_j$ the \emph{demanded data rate} of client~$j$. The \emph{benefit} of client client~$j$ that is associated with AP~$i$ is given by the ratio of ${R_{ij}/{Q_j}}$.
The general objective is to find an assignment that \emph{maximizes} the sum of such benefits, namely the \emph{total benefit} of the network. Therefore, the association problem is modelled by the following linear optimization problem
\vspace{-5mm}\begin{subequations} \label{objfunc_multia}
\begin{align}
\max_{x_{ij}} & \sum\limits_{(i,j) \in C} \frac{R_{ij}}{Q_j} x_{ij} \\
{\rm s.t.}\,\,
& \sum\limits_{j \in A(i)} x_{ij} \geq 1, \quad \forall i=1,\ldots,m, \label{objfunc_multia1} \\
& \sum\limits_{i \in B(j)} x_{ij} = 1, \quad \forall j=1,\ldots,n, \label{objfunc_multia2} \\
& x_{ij} \geq 0, \quad \forall (i,j) \in C
\end{align}
\end{subequations}
The objective function of~\eqref{objfunc_multia} is the total network benefit, where $C$ is the set of all possible AP-client assignment pairs $(i,j)$ (note that $S$ is a subset of $C$) and $(x_{ij})_{j \in A(i)}$ are binary decision variables, indicating the client association. In particular, $x_{ij}=1$ if client $j$ is associated to AP $i$ and $x_{ij}=0$ otherwise, for all $i$ and~$j \in A(i)$. \eqref{objfunc_multia1} and \eqref{objfunc_multia2} ensure that each AP will be assigned to one or more clients and each client will be associated to one AP. Note that from an assignment $S$, we can potentially recover a solution to problem~\eqref{objfunc_multia} by setting $x_{i,j}=1$ if $(i,j)\in S$ and $x_{i,j}=0$ otherwise. An assignment that gives a feasible solution to problem~\eqref{objfunc_multia} is therefore defined as feasible assignemnt. %We will show later in \S~\ref{sec:slol_prob} that problem \eqref{objfunc_multia} can be transformed into a minimum cost flow problem \cite[\emph{Proposition 5.8}]{Bertsekas-98}. Consequently, if the minimum cost flow problem has a feasible solution, then it has an optimal solution where all $x_{ij}$ are either 0 or 1 \emph{as desired}. 
In what follows, we present the proposed solution approach.

\vspace{-0.3cm}
\section{Solution Approach}\label{sec:sol_prob}\vspace{-0.2cm}

The considered problem \eqref{objfunc_multia} is a classical \textit{multi-assignment problem}, where an AP can be assigned to more than one client. Unfortunately, there are no specialized \textit{network flow methods} that can efficiently solve this class of assignment-like problems. There are approaches that apply general purpose network methods such as \textit{primal-simplex}, \textit{primal-dual}, or \textit{relaxation methods}, which may have high complexity ~\cite{Boyd-Vandenberghe-04}. Moreover, general methods for linear optimization, such as the simplex or even interior point methods, do not exploit the particular structure of the considered multi-assigment problem at hand (a network optimization structure) and are not amenable for distributed computation. Thus they are generally less efficient than network optimization methods \cite{Bertsekas-98}. Consequently, we resort to network optimization theory and propose a solution method that combines \textit{auction} algorithms to solve efficiently problem \eqref{objfunc_multia}.

We start by converting problem \eqref{objfunc_multia} into a typical minimum cost flow problem~\cite{Bertsekas-98} by introducing a virtual supernode $s$ that is connected to each AP $i$\footnote{We consider a network where supernode $s$ generates $n-m$ units of traffic and is connected to each AP $i$ by a zero cost arc $(s, i)$. The traffic that is generated at each AP (\emph{supply}) is of one unit. AP $i$ is connected to client $j$ by an arc $(i, j)$ with cost $-R_{ij}/Q_j$.}
\vspace{-2mm}\begin{subequations} \label{objfunc_flow}
\begin{align}
\min_{x_{ij}} & \sum\limits_{(i,j) \in C} \frac{-R_{ij}}{Q_j} x_{ij}  \\
{\rm s.t.}\,\,
& \sum\limits_{j \in A(i)} x_{ij} - x_{si} = 1, \quad \forall i=1,\ldots,m, \label{objfunc_flow1} \\
& \sum\limits_{i=1}^m x_{si} = n - m, \label{objfunc_flow2}\\
& \sum\limits_{i \in B(j)} x_{ij} = 1, \quad \forall j=1,\ldots,n, \label{objfunc_flow3}\\
& x_{ij} \geq 0, \quad \forall (i,j) \in C, \\
& x_{si} \geq 0, \quad \forall i=1,\ldots,n
\end{align}
\end{subequations}
where the sign of the benefit was reversed (cost coefficient) compared to problem \eqref{objfunc_multia}, \emph{minimization} replaced the \emph{maximization} and $x_{ij}$ was extended to include also the supernode $s$. By using the terminology of network optimization, $x_{i,j}$ has the meaning of amount of flow between $i$ and $j$, and the first constraint ensures that the flow \emph{supply} of each AP $i$ is one unit, while the second one declares that $s$ is the source node and the flow that generates is of $n-m$ units. Therefore, a flow of one unit will reach each client $j$. The last two constraints declare that the flow of each arc may be infinite, where an arc between $i$ and $j$ denotes the connection $(i,j)$. A solution to the minimum cost flow problem \eqref{objfunc_flow} is the same to the initial multi-assignment problem \eqref{objfunc_multia}.

By using the duality theory for minimum cost network flow problems \cite[\emph{\S 4.2}]{Bertsekas-98} we formulate the dual problem
\vspace{-4mm}\begin{subequations} \label{objfunc_flowd}
\begin{align}
\min_{\pi_i, p_j, \lambda} & \sum\limits_{i=1}^m \pi_i + \sum\limits_{j=1}^n p_j + (n-m)\lambda  \\
{\rm s.t.}\,\,
& \pi_i + p_j \geq \frac{R_{ij}}{Q_j}, \quad \forall (i,j) \in C, \\
& \lambda \geq \pi_i, \quad \forall i=1,\ldots,m
\vspace{-2mm}\end{align}
\end{subequations}
where $-\pi_i$ is the Lagrangian multiplier associated with constraint \eqref{objfunc_flow1} representing the price of each AP $i$, $\lambda$ is the Lagrangian multiplier associated with constraint \eqref{objfunc_flow2} representing the price of the supernode $s$ (recall that $s$ is the source of the flows), and $p_j$ is the Lagrangian multiplier associated with constraint \eqref{objfunc_flow3} representing the price of each client $j$. The optimal solution to problem \eqref{objfunc_flowd} allows us to derive the optimal solution to \eqref{objfunc_multia} \cite[\emph{\S 4.2, \S 5}]{Bertsekas-98}.

In order to solve problem \eqref{objfunc_flowd} we need some technical intermediate results. We start by giving the definition of $\epsilon$-\textit{Complementary Slackness} ($\epsilon-CS$): Let $\epsilon$ be a positive scalar, we say that an assignment $S$ and a pair ($\pi, p$) satisfy $\epsilon-CS$ if
\vspace{-2mm}\begin{subequations} \label{e-cs}
\begin{align}
& \pi_i + p_j \geq \frac{R_{ij}}{Q_j} - \epsilon, \quad \forall (i,j) \in C,\\
& \pi_i + p_j = \frac{R_{ij}}{Q_j}, \quad \forall (i,j) \in S,\\ \label{e-cs-3}
& \pi_i = \max_{k=1,\ldots,m} \pi_k, \quad \forall i \mbox{\ s.t. $i$ has more than one pair\ } (i,j) \in S
\vspace{-2mm}\end{align}
\end{subequations}

\begin{prop}\label{prop:ecs}
Consider problems \eqref{objfunc_multia} and \eqref{objfunc_flowd}. Let $S$ be a feasible solution for problem~\eqref{objfunc_multia} and consider a dual variable pair ($\pi, p$). Let $\epsilon<1/m$ and assume $R_{ij}/Q_j$ be integer $\forall i,j$. If $\epsilon-CS$ conditions \eqref{e-cs} are satisfied by $S$ and $\pi, p$, then $S$ is optimal for problem \eqref{objfunc_multia}.
\end{prop}
\begin{proof}
The proof is ad-absurdum. Assuming that $S$ is not optimal, then there is a new assignment that can improve the objective function \eqref{objfunc_flowd} and can give us a new solution: Let $E$ be a cycle, namely a collection of arcs that start and end with the same node, that includes also the supernode $s$: $E=(s, i_1, j_2, i_2,..., i_{k-1}, j_k, i_k, s)$. In this solution, the nodes $i_t$ represent the APs, while the nodes $j_t$ represent the clients and $(i_t, j_t) \in S, j_t \in A(i_{t-1}), (i_{t-1}, j_t) \notin S, t=2,...,k$. Based on \emph{max-flow theory} \cite[\emph{\S 3}]{Bertsekas-98}, augmentation along $E$ is achieved by replacing $(i_t, j_t) \in S$ by $(i_{t-1}, j_t)$ in $S$, $t=2,...,k$. AP $i_k$ must be assigned to more than one clients prior to the previous operation because the arc $(i_k, j_k)$ will exit the assignment and therefore, the AP $i_k$ will be left unassigned. This will result to an infeasible solution to problem \eqref{objfunc_flowd}. Moreover, $k\leq m$ since $E$ cannot contain repeated clients. Considering also that $\epsilon<1/m$ we conclude that $k\epsilon<1$.

Since we achieved strict cost improvement in the previous operation, we have
\vspace{-2mm}\begin{equation} \label{ap1}
\sum_{t=2}^k \frac{R_{i_tj_t}}{Q_{j_t}} + 1 \leq \sum_{t=2}^k \frac{R_{i_tj_t}}{Q_{j_t}},
\vspace{-2mm}\end{equation}
In order to reveal the $\epsilon-CS$ conditions \eqref{e-cs}, we transform \eqref{ap1} as
\vspace{-2mm}\begin{equation} \label{ap2}
\sum_{t=2}^k \left(\frac{R_{i_tj_t}}{Q_{j_t}}-p_{j_t}\right) + 1 \leq \sum_{t=2}^k \left(\frac{R_{i_tj_t}}{Q_{j_t}}-p_{j_t}\right).
\vspace{-2mm}\end{equation}
Now using the $\epsilon-CS$ conditions \eqref{e-cs}, \eqref{ap2} can be written as
\vspace{-2mm}\begin{equation} \label{ap3}
\sum_{t=2}^k \pi_{i_t} + 1 \leq \sum_{t=2}^k \left(\frac{R_{i_tj_t}}{Q_{j_t}}-p_{j_t}\right)
\leq \sum_{t=1}^{k-1} \pi_{i_t} + (k-1)\epsilon.
\vspace{-2mm}\end{equation}
From~\eqref{ap3} we have $1-(k-1)\epsilon\leq \pi_{i_1}-\pi_{i_k}$ which contradicts $k\epsilon<1$, because AP $i_k$ is assigned to more than one clients, i.e., $\pi_{i_k}\geq \pi_{i_1}$, [compare with \eqref{e-cs-3}]. We conclude that our first assumption on that $S$ is non optimal is wrong, which implies that $S$ is optimal. We can get similar results considering that supernode $s$ is not part of $E$.
\end{proof}

\begin{figure*}[ttt!]
\vspace{-1cm}\begin{minipage}[t]{3.1in} 
\begin{algorithm}[H] \small
\caption{\small Forward Auction for Client Assignment}
\label{alg_f}
\begin{algorithmic}
\REQUIRE Initial values of $S, p$
\ENSURE $R_{ij}/Q_j - p_j \geq  \max_{k\in A(i)} \left\{R_{ij}/Q_j-p_k\right\} - \epsilon, \quad \forall (i,j) \in S$
\WHILE{there are unassigned clients}
\STATE client $j$ is unassigned in $S$
\STATE find the best client $j_i$ such that:\\
$j_i = \arg \max_{j\in A(i)} \left\{R_{ij}/Q_j - p_j\right\},$\\
$u_i = \max_{j\in A(i)} \left\{{R_{ij}/Q_j - p_j}\right\},$\\
$\omega_i = \max_{j\in A(i), j\neq j_i} \left\{{R_{ij}/Q_j} - p_j\right\},$
\IF {$j_i$ is the only client in $A(i)$}
\STATE $\omega_j \rightarrow -\infty$
\ENDIF\\
\STATE $b_{i{j_i}} = p_{j_i}+u_i- \omega_i+\epsilon = R_{ij}/Q_{j_i}- \omega_i+\epsilon$
%\STATE Add $(i_j,j)$ to $S$:\\
%$p_j=\beta_j-\delta$\\
%$\pi_{i_j}=\pi_{i_j}+\delta$
\STATE $p_j=\max_{i\in P(j)}{b_{ij}}$, where $P(j)$ is the set of APs that client $j$ received a bid,
%\IF {$\delta>0$}
\STATE remove any pair $(i, j)$, where $j$ was initially assigned to some $i$ under $S$, and add the pair $(i_j, j)$ to $S$ with $i_j=\arg \max_{i\in P(j)}{b_{ij}}$
%\ENDIF
\ENDWHILE
\end{algorithmic}
\end{algorithm}
\end{minipage}
\hfill
\begin{minipage}[t]{3.1in}
\begin{algorithm}[H] \small
\caption{\small Reverse Auction for Client Assignment}
\label{alg_r}
\begin{algorithmic}
\REQUIRE $S$, ($\pi, p$) and $\lambda$ from forward auction
\ENSURE (1) $\pi_i + p_j \geq R_{ij}/Q_j - \epsilon, \quad \forall (i,j) \in C$ and
 (2) $\pi_i + p_j = R_{ij}/Q_j, \quad \forall (i,j) \in S$
\WHILE{there are unassociated clients}
\STATE client $j$ is unassociated in $S$
\STATE find the best AP $i_j$ such that:\\
$i_j = \arg \max_{i\in B(j)} \left\{R_{ij}/Q_j - \pi_i\right\},$\\
$\beta_j = \max_{i\in B(j)} \left\{R_{ij}/Q_j - \pi_i\right\},$\\
$\omega_j = \max_{i\in B(j), i\neq i_j} \left\{R_{ij}/Q_j - \pi_i\right\},$
\IF {$i_j$ is the only AP in $B(j)$}
\STATE $\omega_j \rightarrow -\infty$
\ENDIF\\
\STATE $\delta = \min{\{\lambda-\pi_{i_j}, \beta_j-\omega_j+\epsilon\}}$
\STATE add $(i_j,j)$ to $S$: $p_j=\beta_j-\delta$, $\pi_{i_j}=\pi_{i_j}+\delta$
\IF {$\delta>0$}
\STATE remove the pair $(i_j, j_{old})$ where $j_{old}$ was initially assigned to $i_j$ under $S$
\ENDIF
\ENDWHILE
\end{algorithmic}
\end{algorithm}
\end{minipage}
\vspace{-0.8cm}
\end{figure*}

Note that in general, $R_{ij}/Q_j$ is not an integer as required by the Proposition~\ref{prop:ecs}, i.e., rounding those to the closest integer value or scalled to an integer value is necessary before running the algorithms. However, in 60 GHz access networks, the effect of rounding influences slightly the true optimal value of Problem~\eqref{objfunc_flowd}, because we can assume $R_{ij}\gg Q_j$, and as a result the fractional part of $R_{ij}/Q_j$ is relatively smaller than the integer part of $R_{ij}/Q_j$.

Based on Proposition~\ref{prop:ecs}, we are now in the position to present the solution method to problem~\eqref{objfunc_flowd} by an auction mechanism. First, a forward auction algorithm associates each AP to one client, see Algorithms~\ref{alg_f}. Then, a modified reverse auction is applied to assign the rest of the clients to the available APs, see Algorithms~\ref{alg_f}. Finally, we show that the execution of the two algorithms terminates with an optimal solution by a finite number of iterations.

In particular, we start from a feasible assignment $S$ and the corresponding $(\pi, p)$ pair that satisfy the first two $\epsilon-CS$ conditions. We apply Algorithm~\ref{alg_f} until each AP is associated with a single client and until the $\epsilon-CS$ conditions are satisfied. At this stage, some of the clients can still be unassigned. We then apply Algorithm~\ref{alg_r} that gets as input the assignment achieved by Algorithm~\ref{alg_f} ($S$ and $(\pi, p)$). We compute the maximum initial profit for the APs
%\begin{equation} \label{lamda}
$\lambda = \max_{i=1,\ldots,m} \pi_i$.
%\end{equation}
The iterative Algorithm \ref{alg_r} maintains an assignment $S$, where each AP is associated with at least one client, and a pair ($\pi, p$) that satisfies the first two $\epsilon-CS$ conditions. Algorithm \ref{alg_r} terminates when all unassigned clients have been assigned to an AP. While $\lambda$ is kept constant through the execution of Algorithm \ref{alg_r}, \eqref{e-cs-3} is satisfied upon~termination.

\begin{prop}\label{prop:opt}
Consider Algorithm \ref{alg_f} and \ref{alg_r}. Let Algorithm \ref{alg_f} run first and then let Algorithm \ref{alg_r} run iteratively. Algorithm \ref{alg_r} terminates in a finite number of iterations with an optimal AP-client assignment when $\epsilon<1/m$.
\end{prop}
\begin{proof}
In order to prove the optimality and the convergence of the modified reverse auction algorithm we have to show that a) The modified reverse action algorithm iterates by satisfying $\epsilon-CS$ conditions \eqref{e-cs} and $\lambda= \max_{i=1,...,m} \pi_i$, b) The algorithms terminates after a finite number of iterations with a feasible assignment.

The proof of a) is a straight forward application of the well known \emph{theory for auction algorithms} \cite[\emph{\S 7}]{Bertsekas-98} to show that if the $e-CS$ conditions \eqref{e-cs} and $\lambda= \max_{i=1,...,m} \pi_i$ are satisfied at a start of an iteration, they are also satisfied at the end of the iteration.

To show b), we observe that an AP $i$ can receive a bid only a finite number of times after $\pi_i=\lambda$. This is true due to that in each iteration the corresponding client will be assigned to AP~$i$ without changing the association of already assigned clients to AP~$i$ (see Algorithm \ref{alg_r}). Lastly, at the end of each iteration when AP $i$ receives a bid, the profit $\pi_i$ is either equal to $\lambda$ or else increases by at least $\epsilon$. Since $\lambda$ is an upper bound in the profits throughout the algorithm, the main outcome is that each AP can receive a finite number of bids (finite termination).
\end{proof}

\emph{Proposition} \ref{prop:opt} is very important to get an insight of the behavior of the proposed algorithm in general networks, see Section~\ref{sec:numerical_results} for numerical examples.

%\begin{algorithm}[t] \small
%\caption{Reverse Auction for Client Assignment}
%\label{alg_r}
%\begin{algorithmic}
%\REQUIRE $S$, ($\pi, p$) and $\lambda$ that resulted from the forward auction
%\ENSURE (1) $\pi_i + p_j \geq \frac{R_{ij}}{Q_j} - \epsilon, \quad \forall (i,j) \in A$ and
%\ENSURE (2) $\pi_i + p_j = \frac{R_{ij}}{Q_j}, \quad \forall (i,j) \in S$
%\WHILE{there are unassociated clients}
%\STATE client $j$ is unassociated in $S$
%\STATE find the best AP $i_j$ such that:\\
%$i_j = \arg \max_{i\in B(j)} \left\{{\frac{R_{ij}}{Q_j} - \pi_i}\right\},$\\
%$\beta_j = \max_{i\in B(j)} \left\{{\frac{R_{ij}}{Q_j} - \pi_i}\right\},$\\
%$\omega_j = \max_{i\in B(j), i\neq i_j} \left\{{\frac{R_{ij}}{Q_j} - \pi_i}\right\},$
%\IF {$i_j$ is the only AP in $B(j)$}
%\STATE $\omega_j \rightarrow -\infty$
%\ENDIF\\
%\STATE $\delta = \min{\{\lambda-\pi_{i_j}, \beta_j+\epsilon\}}$
%\STATE add $(i_j,j)$ to $S$:\\
%$p_j=\beta_j-\delta$\\
%$\pi_{i_j}=\pi_{i_j}+\delta$
%\IF {$\delta>0$}
%\STATE remove the pair $(i_j, j)$ where $j$ was initially assigned to $i_j$ under $S$
%\ENDIF
%\ENDWHILE
%\end{algorithmic}
%\end{algorithm}
\vspace{-0.3cm}\section{Numerical Analysis}\label{sec:numerical_results}\vspace{-0.3cm}
In this section we present a numerical evaluation of the proposed algorithm in a multi-user multi-cell environment. We compare the proposed solution approach to a)~random association policy and b)~RSSI-based policy, which is the standard association mechanism used in 802.11ad~\cite{802_11ad}~networks.

We define the SNR operating point at a distance $d$ form any AP as
\vspace{-3mm}\begin{equation} \label{eq:SNR}\nonumber
\texttt{SNR}(d) = \left\{ \begin{array}{ll}
  \displaystyle{{P}_{0}}\lambda^2/(16\pi^2N_0W) & d\leq d_0\\
 \displaystyle{{P}_{0}}\lambda^2/(16\pi^2N_0W)\cdot\left({d}/{d_0}\right)^{-\eta} &  \textrm{otherwise}\ .
  %l & \textrm{morning}
   \end{array} \right.
\vspace{-3mm}\end{equation}
We consider circular cells (assumed for simplicity, without loss of generality), as depicted in~Figure~\ref{fig:system_example}. The radius of each cell $r$ is chosen such that $\texttt{SNR}(r)=10$~dB. APs are located such that the distance between any consecutive AP is $D=1.1r$. The clients are uniformly distributed among the cells, and the potential AP-client association is found as pointed out in Figure~\ref{fig:system_example}. 

\begin{figure*}[ttt!]
\vspace{-1.3cm}\begin{minipage}[t]{2.8in}
\begin{figure}[H]
%\centering
\hspace{-1cm}
{\includegraphics[height=0.23\textheight]{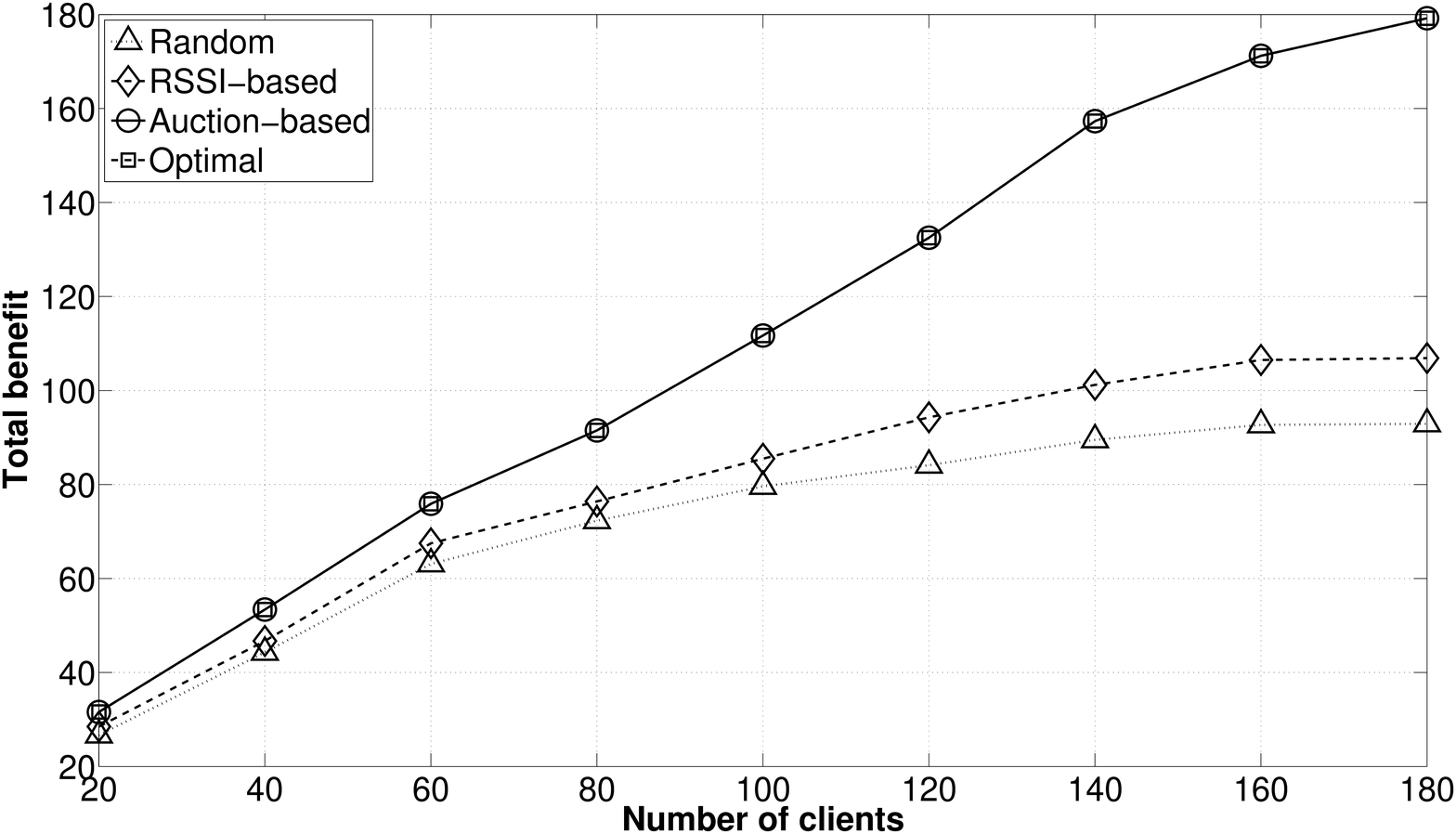}}\vspace{-5mm}
\caption{Total benefit \eqref{objfunc_multia} vs. number of clients (10 APs).}
\label{fig:1}
\end{figure}
\end{minipage}
\hfill
\begin{minipage}[t]{2.8in}\hspace{2cm}
\begin{figure}[H]
%\hspace{3cm}
%\centering
\hspace{-1.5cm}
{\includegraphics[height=0.23\textheight]{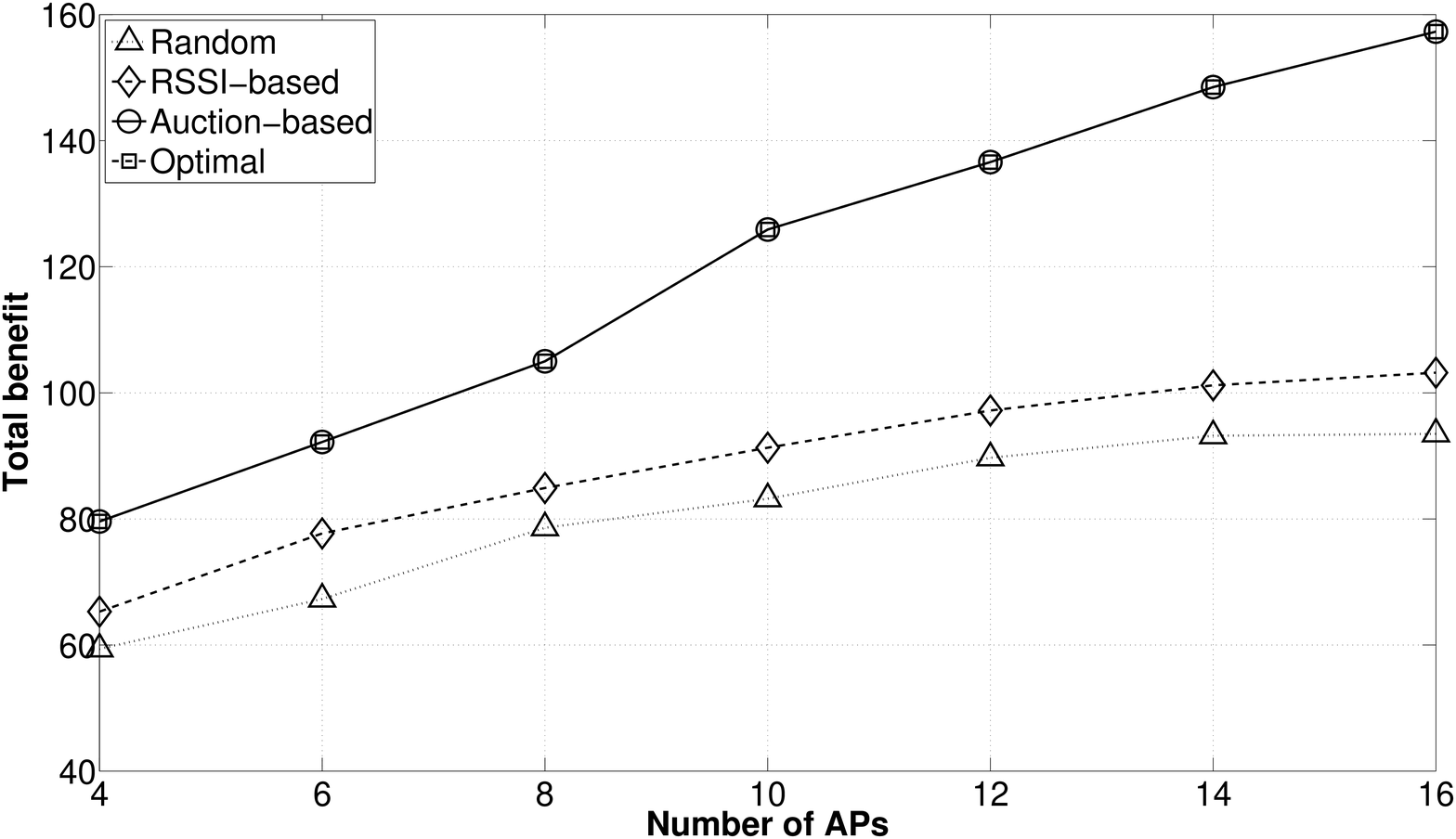}}\vspace{-5mm}
\caption{Total benefit \eqref{objfunc_multia} vs. number of APS (100 clients).}
\label{fig:2}
%\vspace{-0.8cm}
\end{figure}
\end{minipage}
\vspace{-1cm}
\end{figure*}

We set $\lambda=5\:$mm, $N_0=-134\:$dBm/MHz, $W=1200\:$MHz, and $d_0=1\:$m, see~\eqref{eq:AWGN}. Moreover, we set $I_j=0$, $P_{ij}=P_0=0.1\:$mW and $G^{\mathrm{Tx}}_{ij}=G^{\mathrm{Rx}}_{ij}=1$. Furthermore, we assume that $Q_j$s are uniformly distributed on $[0,100]$ Mbits/s. The algorithms were implemented in MATLAB and run on an Intel Core 2 Duo 2.40 GHz processor with 8 GB RAM.

Figure \ref{fig:1} depicts the total benefit $\sum\limits_{(i,j) \in S} R_{ij}/Q_j$ (main objective in problem \eqref{objfunc_multia}) achieved by our solution approach in comparison to the optimal solution of \eqref{objfunc_multia}, the received signal strength (RSSI) based mechanism (adopted by 60 GHz standards) and the random association methodology, where 10 APs are present in the network and the number of the supported clients varies. We observe that the auction-based approach achieves optimal performance and improves the performance of RSSI-based mechanism up to 75\% (especially in high load conditions). 
%Figure~\ref{fig:2} depicts the total benefit in the network when the number of APs that must support 100 clients varies. 
Figure~\ref{fig:2} depicts the total benefit in the network when the number of APs varies for fixed 100 clients.
The behavior of our approach is similar to Figure \ref{fig:1}, evincing its optimal and scalable performance.

Figure \ref{fig:3} shows the performance of the proposed approach for different network sizes (20 clients 2 APs, 40 clients 4 APs, etc) and parameters $\epsilon$. Note that we have $\epsilon<1/m$ for all considered cases, in order to guarantee optimal performance, see Proposition~\ref{prop:ecs}. Results show that the convergence time of the proposed association algorithm is \emph{approximately} linear in $n$.

%We observe that the convergence time of the auction-based association is low even in large networks. The larger $\epsilon$ is, the faster the algorithm becomes. However, in order to guarantee optimal performance we need to ensure that $\epsilon<1/m$, where recall that $m$ is the number of APs.

\begin{figure}\vspace{-0.5cm}
\centering
{\includegraphics[height=0.23\textheight]{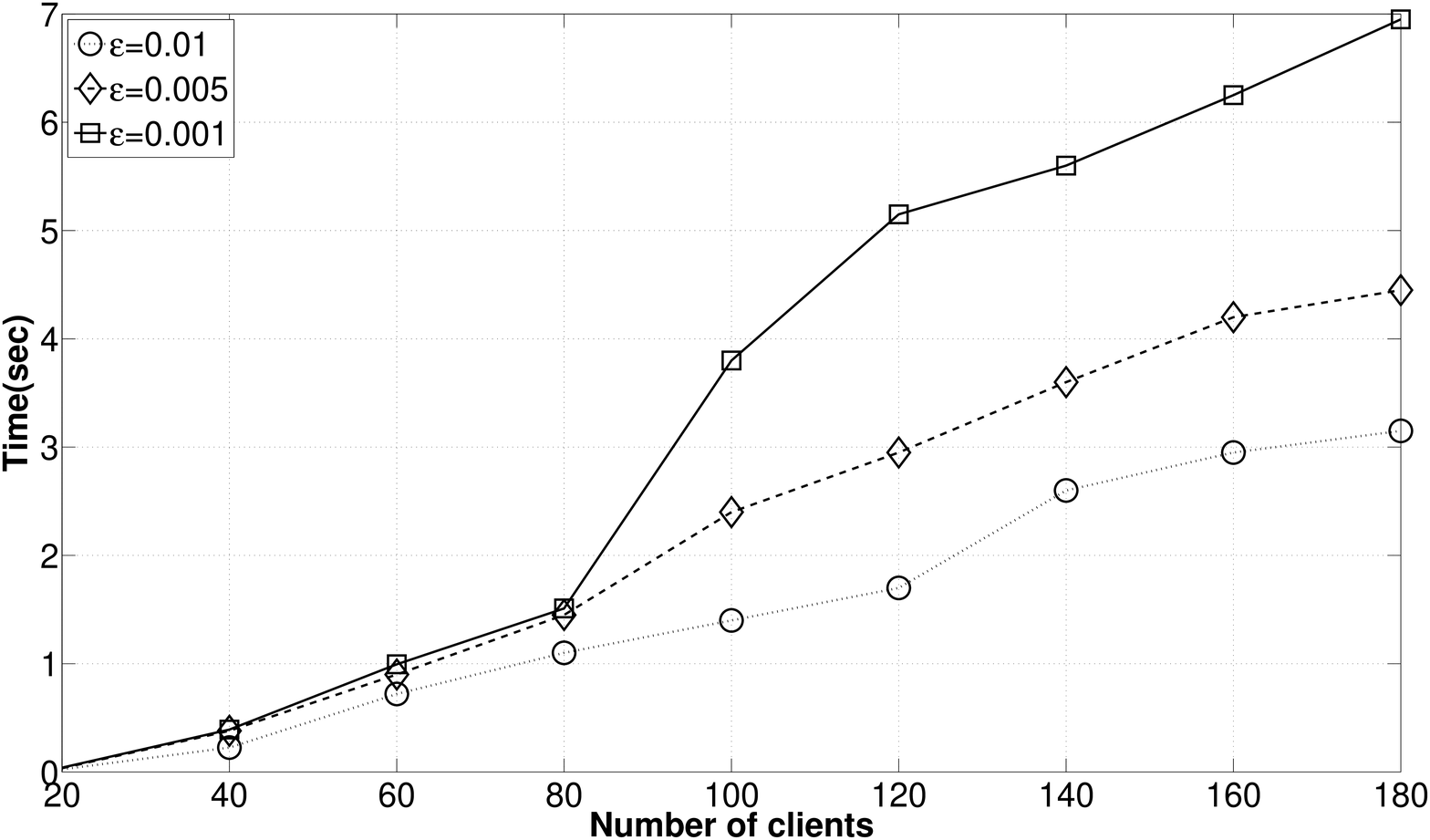}}\vspace{-5mm}
\caption{Time (sec) vs. size of the network.}
\label{fig:3}
\vspace{-1cm}
\end{figure}

\vspace{-0.3cm}\section{Conclusion}\label{sec:conclusions}\vspace{-0.35cm}
We considered the problem of optimizing the allocation of the clients to APs in mmW wireless access networks. The objective in our problem formulation was to maximize the total clients benefit. We presented a solution approach based on forward and reverse auction algorithms. Both theoretical and numerical results evinced the optimal, scalable and time efficient behaviour of our approach. Thus, it could be well applied in the forthcoming 60 GHz wireless access networks.

\ifCLASSOPTIONcaptionsoff
  \newpage
\fi

\vspace{-0.3cm}
\bibliographystyle{IEEEbib}
\bibliography{jour_short,conf_short,references,references_temp}\vspace{-0.5cm}
% that's all folks\vspace{-0.5cm}
\end{document}